\documentclass{llncs}
\usepackage{llncsdoc}

\usepackage{microtype}
\usepackage{mathtools}
\usepackage{amsfonts}
\usepackage[usenames,dvipsnames]{xcolor}
\usepackage[toc,page]{appendix}
\usepackage{tikz}
\usepackage{esvect}
\usepackage{xcolor}
\usepackage{hyperref}

\usetikzlibrary{patterns}
\usetikzlibrary{decorations.pathreplacing}

\newtheorem{observation}{Observation}

\newcommand{\floor}[1]{\left\lfloor #1 \right\rfloor}

\newcommand{\dotpr}[2]{\left\langle #1 , #2 \right\rangle}
\newcommand{\abs}[1]{| #1 |}

\newcommand{\dist}[1]{\lVert #1 \rVert}
\newcommand{\prob}[1]{\mathbb{P}\left[ #1 \right]}

\newcommand{\CNN}{\textit{$c$-approximate nearest neighbor without false negatives}}

\newcommand{\Oc}{\mathcal{O}}
\newcommand{\Rspace}{\mathbb{R}}
\newcommand{\Rdspace}{\mathbb{R}^d}

\newcommand{\sphere}{\mathbb{S}_p^{(d-1)}}
\newcommand{\spheret}{\mathbb{S}^{(d-1)}}

\newcommand{\pfp}{\mbox{$\textsf{p}_{\textsf{fp}}$}}
\newcommand{\tpfp}{\mbox{$\tilde{\textsf{p}}_{\textsf{fp}}$}}
\newcommand{\hpfp}{\mbox{$\hat{\textsf{p}}_{\textsf{fp}}$}}

\title{On fast bounded locality sensitive hashing}

\author{Piotr Wygocki}

\institute{Institute of Informatics, University of Warsaw, Poland\\
  \texttt{wygos@mimuw.edu.pl}}

\authorrunning{P. Wygocki} 

\begin{document}

\maketitle

\begin{abstract}

In this paper, we examine the hash functions expressed as scalar products, i.e., $f(x)=<v,x>$, for some bounded random vector $v$. Such hash functions have numerous applications, but often there is a need to optimize the choice of the distribution of $v$. In the present work, we focus on so-called anti-concentration bounds, i.e. the upper bounds of $\mathbb{P}\left[|<v,x>| < \alpha \right]$. In many applications, $v$ is a vector of independent random variables with standard normal distribution. In such case, the distribution of $<v,x>$ is also normal and it is easy to approximate $\mathbb{P}\left[|<v,x>| < \alpha \right]$. Here, we consider two bounded distributions in the context of the anti-concentration bounds. Particularly, we analyze $v$ being a random vector from the unit ball in $l_{\infty}$ and $v$ being a random vector from the unit sphere in $l_{2}$. We show optimal up to a constant anti-concentration measures for functions $f(x)=<v,x>$. 

As a consequence of our research, we obtain new best results for \newline \textit{$c$-approximate nearest neighbors without false negatives} for $l_p$ in high dimensional space for all $p\in[1,\infty]$, for $c=\Omega(\max\{\sqrt{d},d^{1/p}\})$. These results improve over those presented in \cite{pacukatall}. Finally, our paper reports progress on answering the open problem by Pagh~\cite{Pagh15}, who considered the nearest neighbor search without false negatives for the Hamming distance.

\end{abstract}

\section{Introduction}
Locality sensitive hashing (LSH) functions are hash functions which roughly preserve distance. 
Namely, for two points 'close' to each other in a given metric, the hashes of these points are also 'close' with large probability. 
Analogically, two 'distant' points have 'distant' hashes.
\footnote{In the introduction, we use imprecise terms such as 'close', 'distant', 'small', 'large', etc. in order to avoid introducing complex notation.
These terms are going to be clarified in further sections.}
The concept of LSH is well known and widely used, 
especially in the high dimension nearest neighbor search \cite{DBLP:journals/cacm/AndoniI08,Datar04locality-sensitivehashing,CHAZELLE200824,Pagh15,pacukatall}. 
Normally, one uses LSH to reduce the dimension of a given metric space, usually $l_p^d$ or a Hamming space. 
Common choices of the hash functions are  $f(x) = <x,v>$ or $f(x) = \floor{<x,v>}$, where $v$ is a vector of numbers drawn independently from some probability distribution. 
For instance, the famous Johnson-Linderstrauss Lemma \cite{johnson84extensionslipschitz} can be seen as LSH where $v_i$ 
are independently drawn from the standard normal distribution, for $i\in \{1,\dots d\}$. 
In fact, any distribution with bounded variance produces an LSH function, as $<x,v>$ is a good approximation of $\dist{x}_2$ up to scaling by a constant. 
Such a choice of hash functions has fine theoretical properties. Moreover, they are very cheap to evaluate, which makes them very useful for practical purposes.
The evaluation of a scalar product is proportional to the size of vector representation. 
We say that hash functions with such property are \textit{fast}.
In this paper, we restrict ourselves to such hash functions. 

For the sake of convenience, instead of considering two points $x,y$ 'close' or 'distant', we can consider one point $z = x-y$ and call it 'small' or 'large' respectively.
Given an LSH function, a false positive is a point which is 'large' but its hash is 'small'. Similarly, a false negative is a point which is 'small' but its hash is 'large'.
Naturally, we would like to avoid both false negatives and false positives. 
Many choices of distributions for LSH functions (e.g. normal distribution) give only probabilistic guaranties for both false negatives and false positives. 
Pacuk~et~al.~\cite{pacukatall} considered hash functions where $v$ is a vector of independent Rademacher variables.
Since Rademacher variable is bounded, the hash of a 'small' vector cannot be too 'large'. Consequently, for such a choice of $v$, it is possible to eliminate false negatives. 
The hash functions induced by bounded distributions will also be called bounded.

In this paper, we study the concentration properties of fast bounded LSH functions. 
The crucial concept of this paper is a so-called anti-concentration measure. 
For a given random vector $v$, we are interested in finding the upper bounds of 
$\prob{|<v,x>| \le \alpha}$,  where is $x\in \sphere$.
If variable $X$ is concentrated in $(-A,A)$, say $P(|X| < A) = 1 -\epsilon$, and density of $X$ is symmetric and quasi concave, then  $\prob{|X| \le \alpha} \ge \frac{\alpha}{A}(1 -\epsilon)$. 
We show that the quasi-concaveness is a crucial property of our functions. Actually, the lack of this property was the reason for the inefficiency of the hash functions considered by \cite{pacukatall}.
With the quasi-concaveness assumption on the density function, we show an optimal, up to a constant, fast bounded hash function.

Based on the hash function, we build an algorithm for the \CNN{}. 
In the classical nearest neighbor search, given an input set and a query point, we would like to find a point from input set which is the closest to the query point.
Another variant involves returning any input point\footnote{In practice, we often consider a version of the algorithm which returns all input points within a given radius from the query point. 
Here we consider a one-point query outputs to keep the calculations plain for the reader's convenience. However, all presented results easily transfer to multi-point query outputs. }
within the distance $r$ from the query point, for a given parameter $r$, or reporting that such point does not exist.
Unfortunately, these problems do not have efficient solutions for high dimensional spaces. The existence of such algorithms, with the query and preprocessing complexities not
depending exponentially on the dimension, would disprove the strong exponential time hypothesis \cite{Williams2005}.
In order to overcome this obstacle, we consider the $c$-approximate nearest neighbor, which allows false positives closer than $c r$ to the query point.

As mentioned, known algorithms for the $c$-approximate nearest neighbor give Monte Carlo guaranties. In this paper, we guaranty no false negatives. 
Some known derandomizations result only in theoretical gain since it is easy to tune a probabilistic algorithm to have the exponentially small chance of error 
(e.g. probabilistic prime number testing). This is not true in our case.
Consider a situation where there are many possible result points within the radius $r$ from the query point. 
In such a case, standard LSH algorithms \cite{motwani} need an exponentially large number of hash functions 
to be able to exponentially decrease the chance of a false negative.
In this paper, we improve complexities of the algorithms for the \CNN{} in $l_p$ for all $p\in [1,\infty)$.

The presented algorithms have two stages. In the preprocessing stage, we prepare data structures for further queries. 
In this phase, we use only the input set and the complexity is expected to be polynomial, possibly close to $\Oc(n)$.
In the second stage, we perform the queries. Each query should have the complexity $o(n)$, in order to outrun the trivial full scan algorithm.
In designing the algorithm, we usually need to choose between different configurations of complexities. 
Larger processing time can help reduce the query time and vice versa. In this work, we consider different trade-offs between the query and preprocessing times. 
Improving the hash functions helps us reduce both the query time and the preprocessing time of the \CNN{} for $c = \Theta(\max(\sqrt{d}, d^{1/p}))$ in comparison with the results of \cite{pacukatall}.
Under natural assumptions, we show the hash functions with optimal, up to the multiplicative constant, anti-concentration bounds.

\section{Related Work}\label{rw}


\subsection{The anti-concentration measures}

In this paper we focus on the anti-concentration measures for $<v,x>$, for $x\in \sphere$. 
Let us start with a general bound for functions on a sphere.
Particularly, in the \textit{small ball probability theorem} for some function~$f$ on the unit sphere $\spheret$, we bound $\prob{|f(x)| \le \alpha}$.
The theorem conjectured in \cite{small_ball} and proved in \cite{CORDEROERAUSQUIN2004410} implies that for any Lipschitz function~$f$, with Lipschitz constant $L$,
whose average over the sphere is $1$, we have $\prob{|f(x)| \le \alpha} \le \alpha^{c/L^2}$, for some constant $c$ and $x \in \spheret$.

Carbery and Wright~\cite{4beacfa095ba409292c4365e98dc2dcd} show the following bound for polynomial functions. 
There exists an absolute constant $c>0$ such that, if $Q:\mathbb{R} \rightarrow \mathbb{R}$ is a polynomial of degree at most $k$ and 
$\mu$ is a log-concave probability measure on $\mathbb{R}^m$, then for all $\alpha > 0$:
\[\Big(\int Q^2d\mu\Big)^{\frac{1}{2k}}\mu\{x \in \mathbb{R}^m: |Q(x)| \le \alpha\} \le ck\alpha^{\frac{1}{k}}.\]

Since log-concave probability measures are strongly connected with the surface measure (see Lemma 2  in \cite{small_ball} ), 
the above result gives an alternative way of proving the bounds presented in Section \ref{algo_2}. 
The anti-concentration bound achievable using \cite{4beacfa095ba409292c4365e98dc2dcd}, gives worse constants than the alternative proof provided in this article. 
This is important since this constant is in the exponent of the complexities of the \CNN{} algorithm.

The  anti-concentration measures are strongly connected with the
Littlewood-Offord theory. Consider L\'{e}vy concentration function: 
\[Q(X,\lambda) = \sup_x{\prob{X \le x \le X + \lambda}}.\] 
We have
$\prob{|X| \le \alpha} \le Q(X,2\alpha).$ So any bound on the L\'{e}vy concentration function is also a bound for our problem.
Bobkov~et~al.~\cite{Bobkov2015} considered bounds on the L\'{e}vy concentration function for $X$ being the sum of independent random variables with log-concave density function.
Particularly (Theorem 1.1 in \cite{Bobkov2015}):
\begin{theorem}\label{th_q}
    If $X_1, \dots, X_k$ are independent random variables with log-concave distribution, set $S = \sum_i X_k$.  Then for all $\lambda \ge 0$
    \[Q(S, \lambda) \le \frac{\lambda}{\sqrt{Var(S) + \frac{\lambda^2}{12}}}. \]
\end{theorem}

\subsection{The nearest neighbors}

There exist an efficient $c$-nearest neighbor algorithm for $l_1$ \cite{motwani} with the query and preprocessing complexity equal to $\Oc(n^{1/c})$ and $\Oc(n^{1+1/c})$ respectively  
and a near to optimal algorithm for $l_2$ \cite{DBLP:journals/cacm/AndoniI08}  with query and preprocessing complexity equal to $\Oc(n^{1/c^2+ o(1)})$ and $\Oc(n^{1+1/c^2 + o(1)})$ respectively. 
Moreover, the algorithms presented in \cite{motwani} work for $l_p$ for any $p\in [1,2]$. 
There are also  data dependent algorithms which take into account the actual distribution of the input set \cite{data-depended-hashing}.

Pagh \cite{Pagh15} considered the \CNN{} for the Hamming space, obtaining results close to the results of \cite{motwani}.
Pagh \cite{Pagh15} showed that the bounds of his algorithm for $cr = \log( n/k)$ differ by at most a factor of $\ln 4$ in
the exponent in comparison to the bounds of \cite{motwani}. Indyk~\cite{DBLP:conf/focs/Indyk98} provided a deterministic algorithm 
for $l_{\infty}$ for $c = \Theta(\log_{1+\rho}\log{d})$ with storage $\Oc(n^{1+\rho}\log^{O(1)}n)$ and query time  $\Oc(\log^{O(1)}n)$ for some tunable parameter $\rho$.
Also, Indyk~\cite{Indyk:2007:UPE:1250790.1250881} considered deterministic mappings $l_1^n \rightarrow l_2^m$, for $m=n^{1+O(1)}$, which might be useful for constructing
efficient algorithms for the \CNN{} \cite{Pagh15}.

Eventually the authors of \cite{pacukatall} presented algorithms  for every $p \in [1,\infty]$ and  $c>\tau_p = \sqrt{8}\max \{d^{\frac{1}{2}} d^{1-\frac{1}{p}}\}$.
The considered hash function family is of form $h_p(x) = \floor{<v,x>}$, with the following properties:

\begin{itemize}
 \item \textbf{Close points transform to close hashes}:\newline
 If $\dist{x-y}_p < 1$ then $|h_p(x) - h_p(y)| \le 1$.
 \item \textbf{The probability of false positives}:

For $x,y \in \Rdspace$  such that $\dist{x-y}_p > cr$, it holds:
\[ \pfp = \prob{|h_p(x) - h_p(y)| \le 1} < 1-\frac{(1-\frac{\tau}{c})^2}{2}.\]
\end{itemize}

For such LSH functions the following holds (Theorems 2. and 3. in \cite{pacukatall}):
\begin{theorem}\label{prev}
 For $c>\tau_p = \sqrt{8}\max \{d^{\frac{1}{2}} d^{1-\frac{1}{p}}\}$ and for a large enough $n$, $|P|$ being the size of the result, $p \in [1,\infty]$ we have
the \CNN{} in $l_p$  with the following complexities:
\begin{itemize}
 \item \textbf{for the 'fast query' version:} \newline 
 \begin{itemize}
     \item Preprocessing time: $\Oc(n (\gamma d\log{n}  + (\frac{n}{d})^\gamma))$,
    \item Memory usage: $\Oc(n (\frac{n}{d})^\gamma)$,
    \item Expected query time: $\Oc(d (|P| + \gamma \log(n) + \gamma d))$,
 \end{itemize}
 where $\gamma = \frac{\ln{3}}{-\ln{\pfp}}$ .

\item   \textbf{for the 'fast preprocessing' version:}  \newline
 \begin{itemize}
     \item Preprocessing time: $\Oc(n d\log{n})$,
     \item Memory usage: $\Oc(n\log{n} )$,
     \item Expected query time: $\Oc(d (|P| +  n^{\frac{b}{a+b}}(\frac{b}{a})^{\frac{a}{b+a}}))$,
 \end{itemize}
 where $a = -\ln{\pfp}$, $b=\ln{3}$.
\end{itemize}
\end{theorem}

In this paper, we follow the approach of \cite{pacukatall}. We provide hash functions that satisfy the property of mapping close points to the same values. 
Using the enhanced hash functions we decrease the probability of false positives, which leads to the improvement of the algorithms complexities.
Theorem \ref{main} in the next Section summarizes the obtained results.

\section{Our contribution}
We introduce two classes of hash functions $\hat h_p$ and $\tilde h_p$. 
$\hat h_p$ transforms a given point $x$ to $<v,x>$, where $v$ is a random vector from $l_{\infty}$ ball.
In $\tilde h_p$, we apply the scalar product with a random vector  from sphere $\spheret$.
We prove the anti-concentration bounds for both function families. We follow the schema described in \cite{pacukatall}, which gives the following result:

\begin{theorem}\label{main}
For any  $p\in [1,\infty]$ and for any $c > \tau_p$, we 
        show data structures for the \CNN{} with
\begin{itemize}
    \item $\Oc(n^{1+\frac{\ln{3}}{\ln(c/\tau_p)}})$ preprocessing time  and $\Oc(\log{n})$ query time for the 'fast query' algorithm,
    \item $\Oc(n\log{n})$ preprocessing time  and $\Oc(n^{\frac{\ln{3}}{\ln(3c/\tau_p)}})$ query time for the 'fast preprocessing' algorithm.
\footnote{For simplicity, we omitted the factors dependent on $d$, see \cite{pacukatall} for more details.}
\end{itemize}
We distinguish two cases of the theorem for hash functions $ \hat h_p$ and  $ \tilde h_p$ respectively:
\begin{enumerate}
\item $\tau_p= \hat \tau_p =4\sqrt{3}d^{ \max\{1-1/p, 1/2\}}$,
\item $\tau_p= \tilde \tau_p = 2d^{1/2+|1/2-1/p|}$.
\end{enumerate}
\end{theorem}
The $\hat h_p$ functions give better results for all $p\in[1,2)$, while the $\tilde h_p$ functions work better for $p\in[2,\infty]$. 
Let us now proceed to proving the Theorem \ref{main}. 
We prove case 1. and case 2. in Sections \ref{algo} and \ref{algo_2} respectively.

\section{Definitions}
 The input set will always be assumed to contain $n$
points. In nearest neighbor algorithms, we would like to find points within given distance $r$ from a given query point. 
W.l.o.g, throughout this work we will assume, that $r$ -- a given radius
equals 1 (otherwise all vectors might be rescaled by $1/r$). For $x,y \in \Rdspace$, $<x,y>$ denotes the standard scalar product, i.e.
$<x,y> = \sum_{i=1}^d x_i y_i$.
$\dist{\cdot}_p$ denotes the standard norm in $l_p$, i.e., $\dist{x}_p = (\sum_i{|x_i|^p})^{1/p}$.
$\sphere$ denotes a sphere in $l_p$, i.e., $\sphere = \{x:x \in \Rdspace, \dist{x}_p = 1\}$.
We will write $\spheret$ instead of $\spheret_2$. $U(a,b)$ denote the uniform distribution on the interval $[a,b]$.
The \textit{i.i.d} is the abbreviation for independent and identically distributed.

\section{The algorithm}\label{algo}

The authors of \cite{pacukatall} introduced a general framework for solving the \CNN{} in $l_p$ for any $p \in [1,\infty]$. 
The framework was based on the hash functions $h_p$. Let us recall that
$h_p(x) = \floor{d^{1/p - 1}\dotpr{x}{v}}$,
where $v \in \{-1,1\}^d$ is a random vector satisfying: $\prob{v_i = 1} = 1/2$.
In this section, we will introduce new hash functions $\hat h_p$, which improves over the $h_p$ for $p\in [1,\infty]$.
Particularly, the probability of false positives is decreased, which leads to better complexities of the \CNN{} algorithm for $c = \Theta(d^{\max\{1/2, 1-1/p\}})$.

Given a vector $x \in \Rspace^d$ such that $\dist{x}_p > c$, the probability of a false positive can be bounded as follows \cite{pacukatall}:
\begin{displaymath}   
    \pfp =  \prob{|h_p(x) - h_p(y)| \le 1} < 1-\frac{(1-\frac{\sqrt{8d}}{c})^2}{2}
    .
\end{displaymath}

Even for very large $c$, $\pfp$ is always greater than $1/2$. This must be the
case, since for an arbitrarily large vector $x = (C,C,0,0, \dots, 0)$, the
probability that this vector will be mapped to $0$ equals $1/2$. To overcome
this obstacle, we introduce a new hash function:

\[\hat h_p(x) = \floor{d^{1/p - 1} \dotpr{w}{x}}, \]
 where $w$ is a vector of independent random variables: $w_i \sim U(-1,1)$.

To  bound the probability of false positives, we need to be able to bound the probability of
$\prob{|\dotpr{w}{x}| < \alpha}$: 

\begin{observation}[Anti-concentration bound for a uniform distribution]\label{anti1}
    Let $x\in\Rdspace$ be a fixed vector and $w\in \Rdspace$ be a vector of independent random variables  with $U(-1,1)$ distribution, then
    \[\prob{|\dotpr{w}{x}| < \alpha} \le \frac{2\sqrt{3} \alpha}{\dist{x}_2},\]
\end{observation}
\begin{proof}
To proof this observation, we apply the general bounds  for the L\'{e}vy concentration function for log-concave distributions presented in \cite{Bobkov2015}.
    Let $X_i = w_i x_i$ and $S = \sum_i X_k$. We have
    \[\prob{|\dotpr{w}{x}| < \alpha} = \prob{|S| < \alpha} \le Q(S, 2\alpha).\]
    Since the uniform distribution is log-concave, by applying Theorem \ref{th_q} we get:
    \[\prob{|\dotpr{w}{x}| \le \alpha} \le 
 \frac{2\alpha}{\sqrt{Var(S) + \frac{\alpha^2}{3}}} \le
 \frac{2\alpha}{\sqrt{Var(S)}}.\]
 Since $Var(X_i) = x_i^2/3$ and $Var(S) = \dist{x}_2^2/3$, we have:
    \[\prob{|\dotpr{w}{x}| \le \alpha} \le 
    \frac{2\sqrt{3}\alpha}{\dist{x}_2}.\]

\end{proof}
\qed
If we assume that variables in $w$ are \textit{i.i.d.} and bounded, 
$<w,x>$ satisfy assumptions of the Hoefding inequality \cite{Hoeffding:1963}.
This implies that $<w,x>$ is highly concentrated in the interval $(-|x|_2,|x|_2)S$, where $S$ is the standard deviation of $w_i$. 
Given that, $\hat h_p$ is optimal under the assumption that $w$ are \textit{i.i.d.}.
In order to analyze the properties of the the hash functions, we need the following technical observations:
\begin{observation}\label{bounds}
For any $z\in \Rdspace$ where, $\delta_q =d^{\min\{1/2 -1/q,0\}}$ and $1/p + 1/q = 1$:
\[ \dist{z}_p \delta_p \le \dist{z}_2 \le \dist{z}_p \delta_q^{-1}.\]

\end{observation}
This observation is a direct consequence of the inequality between means.
Given this technical observation and the anti-concentration bound we prove the crucial properties of $\hat h_p$:

\begin{observation}[Close points have close hashes for $\hat h_p$] \label{small_1}
For $x,y \in \Rdspace$, if $\dist{x-y}_p \le 1$ then $\forall_{\hat h_p} |\hat h_p(x) - \hat h_p(y)| \le 1$.
\end{observation}
\begin{proof}
We have:
\begin{displaymath}
    \prob{|\hat h_p(x) - \hat h_p(y)| \le 1} \ge  
    \prob{ \abs{d^{1/p-1} \dotpr{x-y}{w} } \le 1 } 
    .
\end{displaymath}
Since, $\abs{d^{1/p-1} \dotpr{x-y}{v}} \le d^{1/p-1} \dist{x-y}_1 \le \dist{x-y}_p \le 1$, the probability equals~1.
\end{proof}
\qed

\begin{lemma}[Probability of false positives for $\hat h_p$]\label{big_uniform}
For every $p \in [1,\infty]$,
$x,y \in \Rdspace$ and $c > \hat \tau_p = 4\sqrt{3}d^{\max\{1-1/p, 1/2\}}$ such that $\dist{x-y}_p > c$, it holds:

\begin{displaymath}
    \hpfp = \prob{|\hat h_p(x) - \hat h_p(y)| \le 1} < \hat \tau_p /c
    .
\end{displaymath}
\end{lemma}
\begin{proof}
    Let $z=x-y$. We have:
    $$
        \prob{|\hat h_p(x) - \hat h_p(y)| \le 1} \le
        \prob{|\dotpr{z}{w}| \le 2 d^{1-1/p}} \le
        \frac{4\sqrt{3} d^{1-1/p}}{\dist{z}_2}.
    $$
    The second inequality follows from the Observation \ref{anti1}. 
By Observation \ref{bounds}, $\dist{z}_2 \ge \delta_p \dist{z}_p \ge \delta_p c $, which gives:
    $$
        \prob{|\hat h_p(x) - \hat h_p(y)| \le 1} \le
        4\sqrt{3} \frac{d^{1-1/p}}{\delta_p c} 
       .
    $$
    This ends the proof.
\end{proof}
\qed

Theorem \ref{prev} applied to the $\hat h_p$ hash functions results in case 1. of Theorem \ref{main}.
This improves over the complexities presented in \cite{pacukatall}. Particularly,
when $c$ goes to infinity, the preprocessing time in our algorithm tends to $\Oc(n)$, which
was not the case in the preceding algorithm in \cite{pacukatall}.  Still, the
preprocessing complexity is worse than the version which does not give the
guaranties for false negatives: $\Oc(n^{1+1/c})$. This is the price we pay for the
certainty, that all the 'close' points will be found by the algorithm.

\section{The improved algorithm for $p \ge 2$}\label{algo_2}
In this section, we introduce new LSH function family: $\tilde h_p$ which is  tuned up for $p \ge 2$.
We define $\tilde h_p$ as follows:

\begin{displaymath}
    \tilde h_p(x) = \floor{\delta_q \dotpr{w}{x}}, \text{ where }w \text{ is a random vector from the unit sphere }\spheret.
\end{displaymath}

In order to bound the probability false positive, we need to be able to bound the probability of
$\prob{|\dotpr{w}{x}| < \alpha}$. 
We cannot use the techniques introduced in Section \ref{algo}, because random variables in $w$ are not independent.
Instead, the probability can be elegantly expressed
in geometrical terms. $\dotpr{w}{x}$ can be seen as the first coefficient of a random point from $\spheret$. The probability of the complementary event is
proportional to the area of two spherical caps of  distance $\alpha$ from the
origin of $\spheret$. The fraction between the area of
these spherical caps and the area of the unit ball can be expressed as
$I_{\alpha^2}(1/2, (d-1)/2)$ for $|x|_2 =1$, where $I_x(a,b)$ is a regularized
incomplete beta function \cite{sphericalcap}. Bounding the incomplete beta function gives the following observation:

\begin{observation}[The anti-concentration bound for $\spheret$]\label{geom}
Let $x\in \spheret$ be a given unit vector and $w\in \spheret$ be a random unit vector, then
\[\prob{|\dotpr{w}{x}| < \alpha} \le \alpha \sqrt{d} .\]
\end{observation}
\begin{proof}
    As stated before, the complement of the above probability equals the area of two spherical caps of the normalized $(d-1)$-dimensional sphere (i.e. the area of the sphere equals 1).
    For a spherical cap let $ 0 \le \phi\le \pi /2$ denote a colatitude angle, i.e. the largest angle between $e_1$ and a vector from the spherical cap.
    As stated in \cite{sphericalcap}, the area of the spherical cap is given by $1/2 I_{\sin^2{\phi}}((d-1)/2, 1/2)$. 
    Substituting $\alpha = \cos{\phi}$, we have:

  \begin{align*}
  f(\alpha) & = \prob{|\dotpr{w}{x}| < \alpha} =  I_{\sin^2{\phi}}((d-1)/2, 1/2) \\
  & =   I_{1-\alpha^2}((d-1)/2, 1/2) =  I_{\alpha^2}( 1/2, (d-1)/2),
  \end{align*}
  
    where the last equality follows from the fact that $I_x(a,b) = I_{1-x}(b,a)$.
    By the definition of $I_x(a,b)$, we have 
    \[f'(\alpha) = 
    \frac{2\alpha \alpha^{-1}(1-\alpha^2)^{\frac{d-1}{2}}}{B(1/2,(d-1)/2)} = 
    \frac{2(1-\alpha^2)^{\frac{d-1}{2}}}{B(1/2,(d-1)/2)}\]
    and
    \[f''(\alpha) = 
    \frac{-2\alpha(d-3)(1-\alpha^2)^{\frac{d-5}{2}}}{B(1/2,(d-1)/2)},\]

    where $B(a,b)$ is a beta function. For $d=2$ the function $f$ is convex, so 
    \[f(\alpha) \le (1-\alpha) f(0) + \alpha f(1) = \alpha.\] 
    For $d>2$, the function is concave and 
    \[f(\alpha) \le f(0) + \alpha f'(0) = \frac{2\alpha}{B(1/2, (d-1)/2)}.\]

    The last step is proving, that $B(1/2, (d-1)/2) \ge \frac{2}{\sqrt{d}}$. Greni{\'e} et al. \cite{beta} proved that:
\[B(x,y) \ge \frac{x^{x-1}y^{y-1}}{(x+y)^{x+y-1}}.\]
Applying this inequality gives the following bound:
\[B(1/2, (d-1)/2) \ge 
\frac{(1/2)^{-1/2}(\frac{d-1}{2})^{\frac{d-3}{2}}}{(\frac{d}{2})^{\frac{d-2}{2}}} = 
\frac{(1/2)^{-1/2}(\frac{d-1}{d})^{\frac{d-3}{2}}}{(\frac{d}{2})^{\frac{1/2}{2}}} =
\frac{2(\frac{d-1}{d})^{\frac{d-3}{2}}}{\sqrt{d}},\]
which ends the proof, since
$g(d) = (\frac{d-1}{d})^{\frac{d-3}{2}}$ is decreasing for $d\ge3$ and $g(3) = 1$. 



\end{proof}
\qed

For large $d$, $g(d) \approx e^{-1/2}$, what gives a slightly better bound.
Given the above anti-concentration bound we prove the crucial properties of $\tilde h_p$:
\begin{observation}[Close points have close hashes for $\tilde h_p$] \label{small}
For $x,y \in \Rdspace$, if $\dist{x-y}_p < 1$ then $\forall_{\tilde h_p} |\tilde h_p(x) - \tilde h_p(y)| \le 1$.
\end{observation}
\begin{proof}
We have:
\begin{displaymath}
    \prob{|\tilde h_p(x) - \tilde h_p(y)| \le 1} \ge  \prob{ \abs{ \dotpr{x-y}{w} }\delta_q \le 1 }
    .
\end{displaymath}

Applying, in turn, the Schwarz inequality and Observation \ref{bounds} we get:
 \[\delta_q|\dotpr{x-y}{w}| \le \delta_q\dist{x-y}_2 \le \dist{x-y}_p \le 1.\] 
Hence, the points will inevitably hash into the same or adjacent buckets.
\end{proof}
\qed

\begin{lemma}[Probability of false positives for $\tilde h_p$]\label{big_ball}
For every $p \in [1,\infty]$,
$x,y \in \Rdspace$ and $c > \tilde \tau_p = 2d^{1/2+|1/2-1/p|}$ such that $\dist{x-y}_p > c$, it holds:

\begin{displaymath}
    \tpfp = \prob{|\tilde h_p(x) - \tilde h_p(y)| \le 1} < \tilde \tau_p / c
    .
\end{displaymath}
\end{lemma}
\begin{proof}

    Let $z=x-y$ and $X =\dist{z}_2^{-1} \dotpr{w}{z} $, be a random variable. 

    We have:
    $$
        \prob{|\tilde h_p(x) - \tilde h_p(y)| \le 1} \le
        \prob{|X|\dist{z}_2 \le 2 \delta_q^{-1}} \le
        \prob{|X| \le 2(\dist{z}_p \delta_q \delta_p)^{-1}}.
    $$
    The second inequality follows from the Observation \ref{bounds}. 
    Since $\delta_q \delta_p = d^{-|1/2-1/p|}$, we have:
    $$
        \prob{|\tilde h_p(x) - \tilde h_p(y)| \le 1} \le
        \prob{|X| \le 2\dist{z}_p^{-1} d^{|1/2-1/p|}} \le
        \prob{|X| \le 2c^{-1}d^{|1/2-1/p|}}
        .
    $$
    Applying the anti-concentration bound ends the proof.
\end{proof}
\qed

Theorem \ref{prev} applied to the $\tilde h_p$ hash functions results in case 2. of Theorem \ref{main}.
For $p \in [2,\infty]$ we have asymptotically the same constraints on $c$ ($c = \Oc(d^{1-1/p})$).
In addition, for any $p \in [1,\infty]$ we have $\tpfp < \hpfp$. Although the improvement in the bound for $\pfp$ is only in constant, this might be important for
practical cases, because this constant is present in the exponent of the complexities of the \CNN{} algorithm.
For $p \in [1,2)$ there are discrepancies between  the constraints on $c$, depending on the hash functions used. Particularly, the hash functions $h_p$ and $\hat h_p$
work for any $c = \Omega(\sqrt{d})$ for $p \in [1,2)$, while the $\tilde h_p$
works for $c = \Omega(d^{1/p})$.  

A natural approach for optimizing both the probability of false positives and the constraint on $c$ would be 
to consider hash functions of form $\check h_p = \floor{<x,w>}$, where $w$ is a
random point from $\mathbb{S}_q^{(d-1)}$ for $1/q + 1/p = 1$.
\footnote{There are many
possibilities of choosing a random point from a sphere in $l_p$. We
conjecture that the bounds should hold for both geometric surface measure
and cone measure.}  
The H\"{o}lder inequality implies the property of 
'close' points being hashed to adjacent buckets.  In order to prove the bounds for false
positives, we need to bound $\prob{|<x,w>| < \epsilon }$. We conjecture
that this probability can be bounded by $\Oc(\epsilon \sqrt{d})$ for any $p \in
[1,2]$. This is true for $p=2$, since $\check h_2 = \tilde h_2$. Also for large $d$,
$\check h_1 \approx \hat h_1$, because these two functions differ only by the factor of $\max_i{|u_i|}$, where $u_i \sim U(-1,1)$. 
This factor will be close to $1$ for large $d$. Still, techniques used to prove bounds for $\tilde h_p$ and $\hat h_p$ seem to be insufficient to prove
more general bounds for $\check h_p$.

\section{Conclusion and Future Work}

We introduced hash functions $\hat h_p$ and
$\tilde h_p$. Using these functions, we were able to improve the query and
the preprocessing time complexities for the \CNN{} for any $p \in [1,\infty)$.  This is a major
improvement over the results presented in \cite{pacukatall}. 

The future work concerns further relaxing of the restrictions on the approximation factor $c$
and reducing the time complexity of the algorithm or proving that these
restrictions are essential. We wish to match the time complexities given in
\cite{motwani} or show that the achieved bounds are optimal. 

Also, many interesting theoretical problems arise. Consider for instance a random (e.g.,
random in cone measure) point $v$ from $\mathbb{S}_q^{(d-1)}$ and a fixed point $w$ from
$\sphere$ ($1/p + 1/q = 1$, $p\in [1,2)$). A problem can be  posed, whether the 
probability $\prob{|\dotpr{w}{x}| < \epsilon}$ can be bounded. We conjecture, that this
probability is $\Oc(\epsilon\sqrt{d}$).

\section{Acknowledgments}

This work was supported by ERC PoC project PAAl-POC 680912. 

\bibliography{bib}


\end{document}